\documentclass{llncs}
\usepackage{amsmath}
\usepackage{amsfonts}
\usepackage{graphicx}
\usepackage{algorithm}
\usepackage{algpseudocode}
\usepackage{pifont}

\begin{document}

\title{On The Security Evaluation of Partial Password Implementations}



\author{Theodosis Mourouzis\inst{1} \and Marcin Wojcik\inst{2}
 \and Nikos Komninos \inst{3}}

\institute{Cyprus International Institute of Management, Nicosia, Cyprus,\\
\email{theodosis@ciim.ac.cy}\\ 
\and
University of Cambridge,
Computer Laboratory, Cambridge, UK\\
\email{marcin.wojcik@cl.cam.ac.uk}
\and
City, University of London, Department of Computer Science, London, UK,\\
\email{nikos.komninos.1@city.ac.uk}}

\maketitle

\begin{abstract}
A \textit{partial password} is a mode of password-based authentication that is
widely used, especially in the financial sector.
It is based on a challenge-response protocol, where at each login attempt,
a challenge requesting characters from randomly selected positions of
a pre-shared secret is presented to the user.
This mode could be seen as a
``cheap way'' of preventing for example
a malware or a keylogger installed on
a user's device to learn the full password in a single step.
Despite of the widespread adoption of this mechanism, especially
by many UK banks, there is
limited material in the open literature.
Questions like how the security of the scheme varies with the
sampling method employed to form the challenges or what are the existing server-side implementations are left unaddressed.
In this paper, we study questions like how the security of this mechanism
varies in relation to the number of challenge-response pairs available to an attacker
under different ways of generating challenges.
In addition, we discuss possible server-side implementations as (unofficially) listed in different online forums 
by information security experts. To the best of our knowledge 
there is no formal academic literature in this direction and one of the aims 
of this paper is to motivate other researchers to 
study this topic. 
\end{abstract}


\newpage

\section{Introduction}\label{sec1}

The design of a secure and efficient user authentication scheme is one of the major concerns for most
enterprises and organizations. A significant amount of money, time and effort are invested
every year to carry out research in this direction. According
to Cybersecurity Ventures, the U.S Government has invested
more than \$$50$ million over the past four years in Multi-Factor Authentication (MFA) techniques,
aiming to improve a simple password-based authentication scheme~\cite{csecreport}. Additionally, many academic
studies in the past studied extensively the security
and usability of password-based authentication techniques
\cite{angela2,statattack1,jobush,manber,humanident1}.

Despite the fact that several methods of authentication,
such as hardware tokens, biometrics, mouse and keyboard keystroke analytics,
have been developed in
the past few years, a simple password-based scheme is still
the primary mean of authentication for many online services.
This is mainly due to the fact that password-based
 authentication is a cheap, efficient and secure (at least in theory) method of
 authenticating users. As shown in~\cite{angela2}, designing other than simple password-based authentication scheme
might be a very complex task mostly due to the fact that
not only the best security engineering practices, e.g., usability and privacy have 
to be applied, but also the human factor has to be taken into consideration.

The security of password-based systems relies on the user to choose a strong enough
password. If this password is not complex enough, then brute-force or dictionary attacks could potentially breach the security of a system \cite{statattack1}.
Brute-force attacks assume that the distribution of human-chosen passwords
is uniform which is not a practical assumption, as human tend 
to select passwords based on patterns or structures arising from their 
natural language. Relatively recent research
has revealed that this curve (of user-selected passwords) is skewed and more sound mathematical
metrics for the security against guessing attacks using large dictionaries are
presented in \cite{bonneau,bonneau2}.

In addition to the above-mentioned human factor, more sophisticated attacks
using, e.g., malware could be performed. These types of attacks predominantly
exploit various phishing campaigns convincing either directly or by
other means like social engineering approaches the potential victim to
unintentionally install malicious software on the target computing device. Upon infection, the victim's machine
is completely controlled by the attacker who can easily obtain user's passwords in a singe step. 

Researchers have realized this problem, therefore
other identification methods in an attempt to
mitigate single-step disclosure of shared-secret by introducing
time-varying challenges have been proposed~\cite{humanident1,humanident2}.
The partial password scheme is an example of such method where
authentication takes place in the form of challenge-response
pairs, with the challenge requesting a set of characters chosen
randomly from a pre-shared
password. It is considered as a very cheap and effective method
against several attacks that could otherwise compromise
a shared-secret in a single step. It is claimed to be more secure than 
the simple password implementation due to the fact that the size of the responses' space grows in a combinatorial way, 
depending on the implementation. For 
example, for a password of length $n$ and a partial-password implementation requesting $m$ characters out 
of $n$, the 
number of possible responses is $\binom{n}{m}$ if no repetitions are allowed and $n^m$ if repetitions are allowed.

Partial password method is widely deployed in the Banking Sector especially in UK as a part of (at least) 2-factor authentication method ~\cite{FC13paper,bankstas} for authenticating users 
in Internet Banking. It
decreases the probability of success of malware-based attacks since the fraudster cannot really provide to their 
Command-and-Control (CC) server the full password in a single step. Even though, the
fraudsters can sometimes bypass this mechanism by exploiting the weakest link, the human, using HTML injections to modify the page presented to the user and 
request the full password, this scenario is out of the scope of this paper.

In general, all type of attacks applied to the simple password implementations, 
apply also to partial-password implementation schemes. The only difference is 
that the attacker requires more data to launch a successful attack, i.e., intersepting 
more times the authentication handshake in order to either reconstruct the full password 
or get enough data to respond correctly to a new challenge with an overwhelming probability. 
Thus, we have three main type of attacks applied also to partial-password implementations, as follows:

\begin{enumerate}
  \item \textbf{Brute Force}: An attacker uses a computer program or a script that produces all possible password combinations using a fixed alphabet. Then, the attacker tries each password, one by one, until authentication is successful. 
  \item \textbf{Dictionary Attack}: An attacker uses a program or script to try to authenticate by cycling through combinations of common words or using dictionaries based on information related to passwords obtained from compromised servers. 
  \item  \textbf{Key Logger}: An attacker uses a program to track all of a user's keystrokes. 
\end{enumerate}

\textit{Outline of Contributions:} Our motivation is to investigate some open questions~\cite{FC13paper}, such as how security of the partial password scheme varies if challenges are generated using a different method, e.g., allowing the same positions to be requested in the same challenge and how information about user's responses only could be used to speed-up dictionary attacks. The later scenario is
close to the scenario of a hardware-keylogger or to a scenario
where the malware has limited capabilities in terms of intercepting
also the challenge presented to the end-user. 
Considering the fact that half of online users access
their banking account at least twice a week~\cite{guessing1}, there is sufficient information exposed that could be used to launch succesfful attacks. 

In addition, we discuss possible server-side partial password implementations 
as (unofficially) indicated by several information security experts in different online forums \cite{plynt,smartarchitects}. Unfortunately, there is no formal academic literature in this direction and we aim to motivate other 
researchers to work in this direction, as partial password implementations are deployed 
by several major banks in their Internet Banking \cite{FC13paper}. 

This paper is organised as follows. In section 2 related studies are discussed. Section 3 presents the partial password implementations. Section 4 discusses
the security of partial password implementations under different attack scenarios and settings. Finally, section 5 concludes the paper and gives future research directions in the field.

\section{Related Studies}\label{sec2}
In this section we present related studies that fall into partial password mode of
authentication. For example, we describe hardware keyloggers attack scenarios in which
an attacker has data related to the responses but nothing
related to the associated challenges. 

\textit{Hardware Keyloggers:}
A paper by Goring~\textit{et al.}~\cite{CSGoring} studies the case of
a hardware keylogger attack, where the attacker can obtain responses
but not challenges. However, their method is limited to a very particular
case where whenever authentication fails, the server presents again
the same challenge to the user. This potentially allows the attacker
to construct challenge-response pairs by just repeating the authentication process.
In this paper, we further investigate this attack model and we study how we can
use data obtained in a keylogger setting combined with large-dictionaries of 
user-selected passwords in order to
speed up dictionary attacks.

\textit{Partial Password Schemes:}
Another paper by
Aspinall~\textit{et al.}~\cite{FC13paper} studies the
security of a particular partial password implementation,
where the positions requested in the challenges are chosen uniformly at random
without replacement. Furthermore, they study how the security of the system
is related to the number
of challenge-response pairs that the attacker has obtained
(defined in~\cite{FC13paper} as \textit{recording} attacks).
In order to speed-up their attacks they applied frequency analysis of letters of user-selected 
passwords, as appearing in the RockYou dataset~\cite{dataset}. In this paper, we study
a more generic scheme in which the challenges are chosen uniformly at
random and repetitions of positions is allowed. This is claimed to be a more
complex scenario and left as future work in ~\cite{FC13paper} and this is the major contribution of 
this paper.


\section{Partial Password Implementation}\label{sec3}


\subsection{Protocol Description}

A partial password is a challenge on a subset of characters from
a full password.
The overall protocol consist of two phases which could be described as follows~\cite{FC13paper}:

\textbf{A. Registration Phase:} The
user selects a password $p=p_0p_1...p_L$ of a desired length and
usually on a restricted alphabet.


\textbf{B. Login Phase:} The authentication phase is based on the
following challenge-response protocol.

\begin{enumerate}
  \item \textbf{Challenge}:
  The server selects a subset of $m$ integers $i_1,i_2,...,i_m$ from the set $\{0,1,2,..,L\}$ and presents the challenge $(i_1,i_2,...,i_m)$ to the user.

 \begin{center}
\begin{tabular}{ l || c | r | r| r| r| r| r| r|}
Index &  1 & 2 & 3 & 4 & 5 & 6 & 7 &  8   \\
User Password &  p & a & s & s & w &  o &  r&  d  \\
\hline
Challenge &    & 2 &  &  & 5 &  &  &  8     \\
Response  &   & a &  &  & w &   &  &  d \\
\end{tabular}
\end{center}

  \item \textbf{Response}: The response will be of the form $(a_1,a_2,...,a_m)$. The user passes this step only if $a_j=p_{i_j}$ for all $1 \leq j \leq m$.
\end{enumerate}

\noindent If the user's response is not correct, then either
the same or a fresh challenge is presented, while on a subsequent
login trial a fresh challenge is generated in
case of a previous successful authentication. The scenario where the same challenge is presented to the user was studied in~\cite{CSGoring}.

In addition, Aspinall and Just studied the security of the scheme when the integers $i_1,i_2,...,i_m$  are
chosen uniformly at random but without replacement \cite{FC13paper},
while the scenario of repetitions allowed is left as open question as it is 
considered more complex. One of the major contributions of this paper is that we study also this scenario.


\newpage

\subsection{Server-side Implementations}
\label{ssimpelm}


In classical password implementations only the hash of the password is enough to be stored on the server. 
Finding a message for a given hash value (or two different messages with the same hash value) for secure cryptographic hash functions is considered computationally 
hard, thus even an adversary with unrestricted access to the 
hash values cannot deduce the password from the hashes, if a secure cryptographic hash functions is employed, such as SHA-256. 

However, in partial password schemes, a new level of complexity in both storage and validation of the shared-secret on the server side is introduced. 
It is not enough anymore to store the hash of the full password and hence standard password hashing schemes do not apply. 
Instead one has to either store the password in a plaintext, or the hashes of different combinations of each password~\cite{plynt,smartarchitects}. 
For the latter solution, it is not trivial to store the hashes of all the combinations of variable length passwords. 
Possibly, this is the reason why most banks are restricting both the length and the alphabet of the user passwords and only request for up to 4 (maximum) different characters in their partial-password implementation schemes~\cite{FC13paper}.

To the best of our knowledge, there is no formal academic literature discussing the problem of server-side implementation of partial password authentication mechanisms. 
We would like to motivate academic research in this direction as those schemes are widely deployed by major banks around the globe.  
Based on our research findings by searching several online security related forums we have indicated that possible implementations deployed in industry might be as follows~\cite{plynt,smartarchitects}:

\begin{enumerate}
  \item The password is stored in plaintext~\cite{plynt}. 
  This imposes a significant risk from a security point of view as an administrator is likely to have a direct access to the password in plaintext form. 
  Furthermore, if the database is compromised then an adversary has access to all plaintext passwords. 
  This solution might be also not complied with policies requiring hashed or encrypted password storage.
  \item The hashes of all possible combinations of letters are stored per password per user~\cite{plynt,smartarchitects}.
  In a general case, where there are not many constraints applied on a password, this solution might lead to significant database issues in terms of required storage space.
  However survey conducted in~\cite{FC13paper} showed that many banking online systems, that are based on the partial password mode of authentication, impose more or less rigorous restrictions on the length of the password and the size of a character set. 
  In extreme cases password could be restricted to a size of only four characters allowing a character set of size 10 (PIN case \cite{FC13paper}). 
  By applying such restrictions, database storage issues become less demanding and thus this extensive hashing method is more applicable in practice. 
  Under this setting, for a password of length $n$ and a partial password scheme that requests $m$ positions we need to store $\dbinom{n}{m}$ possible hashes, which is translated to $l \times \dbinom{n}{m}$ bits of information per user, if a $l$-bit hash function is employed, e.g., $l$=256 for SHA-256.  

\end{enumerate}


\noindent Another practical implementation that one can think of is the following:

\begin{enumerate}
\setcounter{enumi}{2}

\item The password could be stored on the server in an encrypted form with a use of some symmetric-key scheme, like AES. 
In this case, to mitigate any practical key management issues, keys could be managed by a tamper-resistant hardware, i.e., Hardware Security Module (HSM) or a separate authentication server with employed appropriate access control systems in order to avoid unauthorized users to access the cryptographic key. 
This would provide a black-box interface for encryption and substring verification such that when the password characters are passed to the application they are fed into the HSM or the 
authentication server along with the encrypted password. 
The HSM could then decrypt the password and confirm (or reject) the validity of the provided characters. 
However, the drawback of this method is that during authentication, the full password is decrypted and under certain circumstances leakage of this fully decrypted password could occur. 
\end{enumerate}

Considering the survey conducted in~\cite{FC13paper}, there are surprisingly many tight constraints imposed on passwords used in partial password schemes, i.e., the size of acceptable alphabet and length of the password are relatively small, as well as the number of requested characters in the challenges. 
In the case of Internet Banking authentication, most banks request a password within a given range 
and restricted to a given alphabet, usually the alphanumeric of size 36 or numeric of size 10 characters (PIN).

\section{Security Analysis}

In this section we focus on questions like how many challenge-response pairs are sufficient to reconstruct the shared-secret and how many are needed
in order to guess correctly the next challenge in a partial password protocol with sufficiently high probability.

We essentially study the following three attack scenarios:

\begin{itemize}
  \item \textbf{Recording Attacks:} A malware or a keylogger installed on the user's device is recording several (challenge,response) pairs which are sent to the fraudster's server. The main goal of the fraudster is to reconstruct the password.
  \item \textbf{Next-Challenge Attacks:} Same setting as in recording attacks but in this scenario the attacker would like to know the success rate of
      providing the correct response given some pairs.
  \item \textbf{Attacks With Unknown Challenges:} The attacker runs a dictionary attack and for some reason has only a set of responses, without necessarily knowing the corresponding positions. The idea is to examine if such limited information could benefit a lot a dictionary attack. Since human-selected password distribution is known to be skewed \cite{bonneau,bonneau2} this could be seen as another confirmation of this empirical result.
\end{itemize}

In order to tackle the scenario
where the positions in the challenge could repeat in the same
challenge, we resemble the definition of a \textit{multiset} (cf.~Definition~\ref{def1}).

\begin{definition} \label{def1}
A \textbf{multiset} is a 2-tuple $(A,m)$ where $A$ is some set and \linebreak $m:A\rightarrow \mathbb{N}$
a function from $A$ to the set $\mathbb{N}$.
\end{definition}

\noindent The number of multisets of cardinality $k$, with elements taken
from a finite set of cardinality $n$, is
called the \textit{multiset coefficient}. This number is denoted by $\left(\!\!{n\choose k}\!\!\right)$ and is given
by $\binom{n+k-1}{k}$.



\subsection{Recording Attacks} \label{sec21}

Suppose that the user has agreed on a password
$P=p_0...p_L$ of length $L+1$ with $p_i \in \mathcal{A}$, $\forall 1 \leq i \leq L$, where $\mathcal{A}$ the pre-defined alphabet. We have evaluated the security of partial password implementation in
two different scenarios.

\begin{enumerate}
  \item \textbf{Scenario A (Without Replacement):} The challenge is of the form \linebreak $(i_1,i_2,...,i_m)$
       with $0\leq i_j \leq L$, for all $1 \leq j \leq m$
       and $i_{k'} \neq i_k$ for all $1 \leq k',k \leq m$.
  \item \textbf{Scenario B (With Replacement):} The challenge is of the form \linebreak $(i_1,i_2,...,i_m)$
       with $0\leq i_j \leq L$, for all $1 \leq j \leq m$.
\end{enumerate}

Consider the case where malware, installed on the user computing device, is capturing the responses of the user before the
HTTP POST being encrypted with SSL and sends these responses to
the Command-and-Control server. Then, the threat scenario is
that after sufficient data the attacker would be able either
to reconstruct the full password or have a sufficiently high
probability to response correctly to fresh challenges. The security
analysis of both scenarios is based on Theorem~\ref{theorem1}.

\begin{theorem} \label{theorem1}

Let $X$ the number of different positions of the password that the malware
posses after capturing $k$ challenge-response pairs.
The probability \linebreak $p_k(X=i)$, that the malware knows
exactly $i$ out of the total $L+1$ positions is given by Equation~\ref{eq1} and $2$ for
Scenario~A and~B respectively,

\begin{equation}\label{eq1}
 p_k(X=i)= \begin{cases}
\frac{1}{\binom{n}{m}}\sum_{j=0}^{m}\binom{i-j}{m-j}\binom{n-(i-j)}{j}p_{k-1}(X=i-j) &\text{$m \leq i \leq n, k \geq 1$}\\
1 &\text{$i=k=0$} \\
0 & \text{otherwise}
\end{cases}
\end{equation}
\begin{equation}\label{eq2}
 p_k(X=i)= \begin{cases}
\frac{1}{\left(\!\!{n\choose m}\!\!\right)}\sum_{j=0}^{m}\left(\!\!{i\choose m-j}\!\!\right)\binom{n-(i-j)}{j}p_{k-1}(X=i-j) &\text{$1 \leq i \leq n, k \geq 1$}\\
1 &\text{$i=k=0$} \\
0 & \text{otherwise}
\end{cases}
\end{equation}
\end{theorem}

\begin{proof}
If at step $k-1$, the
malware obtained $i-j$ distinct indices and the aim is exactly $i$ by having
another pair, this implies we need to select exactly $j$ from the $n-(i-j)$
unseen ones and select the rest $m-j$ depending on the scenario.
For scenario A, we choose $m-j$ out of the already known $i-j$ indices,
while for B we choose $m-j$ from $i$ indices, allowing repetitions.
\end{proof}

Figure~1 presents how probability varies
against the number of challenge-response pairs. As we
observe, in case of $L+1=8$ and $m=3$, an attacker
can reconstruct the password with probability higher that $70\%$
after recording $7$ pairs in Scenario~A and $11$ for Scenario~B. For $L+1=12$
and $m=3$, $14$ pairs are needed in Scenario~A while $17$ for
Scenario~B for a success
probability $75\%$.



\subsection{Next-Challenge Attack} \label{sec23}

Another question of significant interest is the probability
to respond correctly to a new challenge given $k$ pairs. Denote
these probabilities as $p^A_{k+1}$ and $p^B_{k+1}$ for Scenario~A and~B respectively. Then, we have the following:

\begin{equation}\label{eqq1}
  p^A_{k+1}(i)=\frac{\binom{i}{m}}{\binom{n}{m}},\;\; p^B_{k+1}(i)=\frac{\left(\!\!{i\choose m}\!\!\right)}{\left(\!\!{n\choose m}\!\!\right)}.
\end{equation}

\noindent After $k$ runs, if the attacker knows $i$ positions, the expected number of pairs learned is given by $E^A_{k}$ and $E^B_{k}$ respectively,

\begin{equation}\label{eqq1}
  E^A_{k}=\sum_{i=m}^{n}p_{k}(X=i)\cdot p^A_{k+1}(i),\;\;E^B_{k}=\sum_{i=1}^{n}p_{k}(X=i)\cdot p^B_{k+1}(i).
\end{equation}


In Figure~2 we observe that an attacker has probability higher than $75\%$ to correctly
reply to the next challenge, by having 8 pairs in Scenario A
or equivalently 9 pairs in Scenario B, for $L+1=10$ and $m=3$. Thus,
security of both schemes is similar for average passwords regarding guessing the
next challenge.

\subsection{Attacks With Unknown Challenges } \label{sec22}

In this section, we study the scenario where an
attacker has obtained some information regarding
user's responses, but has no
knowledge to which challenge they correspond. This
is similar to the hardware keyloggers scenario
as mentioned by Goring \textit{et al} in \cite{CSGoring}.
We call this scenario as the
``attacks with unknown challenges'' scenario.

We have experimentally
demonstrated that in case of a dictionary attack if information available from
keyloggers is used, then we have a significant reduction in the dictionary size,
ending up with a reduced number of candidates. This confirms
even more the claim that the probability distribution of
human-selected passwords is skewed \cite{bonneau,bonneau2}. This is due to
the fact that even with having a set of characters randomly selected from a word,
we can limit down tremendously the number of possible candidates in
a dictionary attack. In our experiments, we used as a dictionary the well-studied
RockYou dataset and results are presented in Table 1.

Denoting by $S_P$ the set of available characters
corresponding to a target password $P$ (i.e., for
$P="password"$ $S_P=\{p,a,s,w,o,r,d\}$), we have performed the following three experiments:

\begin{enumerate}
\item \textbf{Experiment A}: $S_P$ and two characters' positions of the password are known.
\item \textbf{Experiment B}: $S_P$ and the length of the password are known
\item \textbf{Experiment C}: $S_P$, the length of the password and two characters' positions are known.
\end{enumerate}

The algorithm we employed to filter down possible password candidates is described in Algorithm~$2$. Note that $R$ is a parameter which is used in order to search for passwords which are close to the length of password $x$ up to a desired margin. In our case we study the scenario
$R=1$, i.e targeting passwords of known length. Table~$1$ presents some
of the results of our experiments.

\newpage 

\begin{algorithm}[!htbp]
\caption{Dictionary-Filter($S_P$,dictionary $D$, $L+1$,$R$)}
\label{CHalgorithm}
\begin{algorithmic}[1]
\State Initialize an empty list $L_D$
\For{each $x$ \Pisymbol{psy}{206} $D$ }
\State Compute $S_x$, the set of distinct character appearing in the word $x$
\State Compute $A=S_P \cap S_x$
\State \textbf{Experiment A:}
\State \textbf{if} $S_P \subset A$ and $(x_i,x_j)=(P_i,P_j)$ known:
\State                  $  x\rightarrow L_D$
\State \textbf{Experiment B:}
\State \textbf{if} $S_P \subset A$ and $|x|=R.(L+1)$:
\State                  $  x\rightarrow L_D$
\State \textbf{Experiment C:}
\State \textbf{if} $S_P \subset A$ and $|x|=R.(L+1)$ and $(x_i,x_j)=(P_i,P_j)$ :
\State                  $  x\rightarrow L_D$
\EndFor
\end{algorithmic}
\end{algorithm}

\begin{table}[!htbp]
\begin{center}
\begin{tabular}{l|{c}|r||r|r|r|r|}
Password  &$S_x$  & $R$  & Experiment A & Experiment B & Experiment C   \\
\hline
password &$\{a,d,o,p,r,s,w\}$  & 1.0  & 2456& 36 & 12  \\
baseball  &$\{a,b,e,l,s\}$ & 1.0&1435&39 &1  \\
dragon &$\{a,d,g,n,o,r\}$ & 1.0  &3378 &29  &3\\
admin  & $\{a,d,i,m,n\}$ &  1.0  &  3695& 17 & 7 \\
querty & $\{e,q,rt,u,y\}$&1.0   &381&4 &1 \\
\end{tabular}
\caption{The number of possible password candidates.}
\end{center}
\end{table}

From Table 1 we can observe that by knowing the set of distinct characters we can speed up the dictionary attack tremendously.
This is expected to happen since humans tend to select words
from their natural language and thus the distribution of
possible $n$-grams follow a certain distribution.
In our future work we plan to study how the number of
possible candidates varies with $R$, i.e., the attacker
posses a fraction of the password's characters. This would be
complex to implement and study.

\newpage

\section{Conclusion}\label{sec:Conclusion}

Partial passwords is a mode of authentication which
is widely deployed by the industry and especially in
UK banking sector~\cite{FC13paper}. It was proposed as a countermeasure
against attacks that could reveal a shared secret in a single step~\cite{humanident1,humanident2}. It is
a challenge-response protocol, where the challenge is of the form, ``\textit{What are the characters of your password
at positions 1,5 and 9 ?}''.

In this paper, we extend the work of Aspinall~\textit{et al.}~\cite{FC13paper}, and
study some of the open questions stated in the same paper. We investigate and
compare the security of several partial password implementations in which the elements in the challenges are generated uniformly at random but without replacement against the one where the replacements are allowed. The latter cases
seems to be more secure, especially for attackers aiming to fully reconstruct the password. They also benefit from simpler implementation since we don't need
to check if the next positions in the challenge were already asked.

Finally, we study the scenario where the attacker has access to responses but not challenges
and whether this information is valuable to dictionary-type attacks. We have experimentally demonstrated that such information can tremendously reduce the number of potential password candidates from a given
dictionary and this confirms again the claim
that the probability distribution of human-chosen secrets is
skewed \cite{bonneau,bonneau2}.

\textit{Further Work:} There are several areas that we would like to investigate
in more details. We would like to extend the \textit{hardware
keylogger attack} to
other scenarios like having a percentage, $p$, of characters from the password, how this $p$ affects the number of possible candidates from the dictionary. In addition,
we plan to explore the usability of partial passwords
which is still not studied despite the wide practical adoption of such mechanisms~\cite{usabilityOB}.

\begin{figure}[!htbp]
\centering
\includegraphics[scale=0.36]{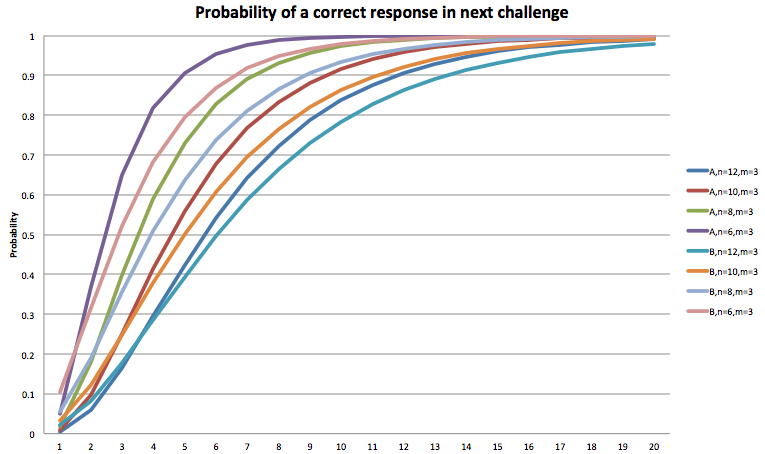}\caption{Probability $p_k(X=i)$ against the number of runs $k$.}
\end{figure}

\begin{figure}[!htbp]
\centering
\includegraphics[scale=0.36]{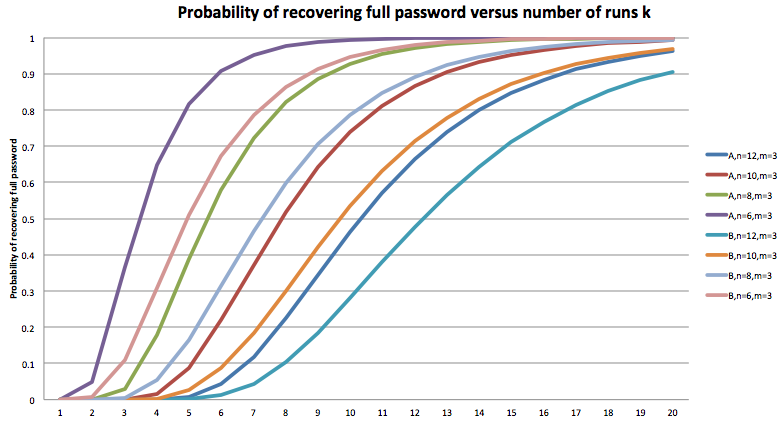}\caption{Expected number of $m$-tuples learned after $K$ runs.}
\end{figure}

\end{document}